\documentclass[journal]{IEEEtran}

\ifCLASSINFOpdf
\else
   \usepackage[dvips]{graphicx}
\fi
\usepackage{url}
\usepackage{graphicx}
\usepackage{amsmath,amsfonts, amssymb}
\usepackage{xcolor}
\usepackage{bbm}
\usepackage{hyperref}

\newtheorem{theorem}{Theorem}

\newtheorem{remark}{Remark}
\newtheorem{definition}{Definition}

\DeclareMathOperator*{\argmin}{arg\,min}

\newcommand{\Min}{{\mathrm{Min}}}
\newcommand{\reg}{{\mathrm{reg}}}
\newcommand{\HT}{{\mathrm{HT}}}
\newcommand{\HE}{{\mathrm{HE}}}

\begin{document}
\title{Optimal Regret Exponents for Bayesian Statistical Decision Problems}

\author{
    Hyun-Young~Park,~\IEEEmembership{Graduate Student Member,~IEEE,}
    and Si-Hyeon~Lee,~\IEEEmembership{Senior Member,~IEEE}
    \thanks{
        This work was supported 
        in part by IITP (Institute of Information \& communications Technology Planning \& Evaluation) under 6G·Cloud Research and Education Open Hub (IITP-2026-RS-2024-00428780) grant % Funding 1 (Cloud)
        and in part by IITP through the Next Generation Semantic Communication Network Research Center (RS-2024-00398948) grant, % Funding 2 (Semantic)
        both funded by the Korea government (MSIT).
    }
    \thanks{
        The authors are with the School of Electrical Engineering, Korea Advanced Institute of Science and Technology (KAIST), Daejeon 34141, South Korea (e-mail: phy811@kaist.ac.kr, sihyeon@kaist.ac.kr).
        \emph{(Corresponding author: Si-Hyeon Lee.)}
    }
}

% \markboth{Journal of \LaTeX\ Class Files, Vol. 14, No. 8, August 2015}
% {Shell \MakeLowercase{\textit{et al.}}: Bare Demo of IEEEtran.cls for IEEE Journals}
\maketitle

\begin{abstract}
We study finite-state finite-action Bayesian statistical decision problems. While exact error-exponent characterizations are known for several special cases, including hypothesis testing and hypothesis exclusion, the asymptotic behavior of the optimal Bayes regret is largely unknown for general decision problems. In this paper, we show that the optimal regret always decays exponentially fast and characterize its exact exponent for arbitrary loss functions. The exponent is given by the minimum multivariate Chernoff information over the minimal incompatible subsets of states, where an incompatible subset is a collection of states for which no single action is optimal for all states in the subset. Our result recovers the classical pairwise-minimum Chernoff exponent for symmetric multiple hypothesis testing and the multivariate Chernoff exponent for hypothesis exclusion, while also yielding, to the best of our knowledge, the first exact exponent characterization for list hypothesis testing.
\end{abstract}


\IEEEpeerreviewmaketitle

\section{Introduction}
\IEEEPARstart{O}{ne} of the fundamental results in information theory is the characterization of the error exponent for hypothesis testing \cite{coverElementsInformationTheory2006, hellmanProbabilityErrorEquivocation1970, kanayaAsymptoticsPosteriorEntropy1995, nielsenInformationGeometricCharacterizationChernoff2013}. In symmetric Bayesian hypothesis testing, one observes i.i.d. samples generated according to one of several candidate distributions and seeks to identify the true hypothesis with the smallest possible error probability. It is well known that the optimal error probability decays exponentially fast with the number of samples, and that the corresponding exponent is given by the minimum pairwise Chernoff information \cite{hellmanProbabilityErrorEquivocation1970, kanayaAsymptoticsPosteriorEntropy1995, nielsenInformationGeometricCharacterizationChernoff2013}. The resulting characterization has become a fundamental tool in statistical signal processing \cite{leeGeneralizedChernoffInformation2012, qiComplexityOutageProbability2013, chenModelingInformationFlows2012}. 
A recently studied counterpart is hypothesis exclusion \cite{mishraOptimalErrorExponents2024, jiConverseBoundsQuantum2025, cavesConditionsCompatibilityQuantumstate2002, barrettNoensuremathpsiEpistemicModel2014}, also known as state anti-distinguishability or state elimination. Rather than identifying the true hypothesis, the objective is to output a hypothesis that is not the true one. This problem has attracted renewed attention due to several applications in quantum information theory \cite{cavesConditionsCompatibilityQuantumstate2002, barrettNoensuremathpsiEpistemicModel2014}. Surprisingly, despite their opposite objectives, hypothesis exclusion exhibits an asymptotic behavior similar to that of hypothesis testing: the optimal error probability again decays exponentially fast, with exponent given by the multivariate Chernoff information of all hypotheses \cite{mishraOptimalErrorExponents2024}.

Both problems can be viewed as special cases of the classical framework of statistical decision theory \cite{bergerStatisticalDecisionTheory1985, dulekOptimalDecisionRules2020}. In a statistical decision problem, the goal is to choose an action after observing the samples so as to minimize the expected loss.
Unlike hypothesis testing, where each state has a unique optimal action, a general decision problem may exhibit a richer optimality structure: a state may admit multiple optimal actions, and a single action may be optimal for multiple states. Hypothesis exclusion is one example of such a many-to-many structure, while in general the pattern of state-action optimality relations can be arbitrary. 
Such situations naturally arise when exact state identification is unnecessary.
A notable example is list hypothesis testing \cite{asadikangarshahiMinimumProbabilityError2023}, where the decision rule outputs a list of candidate states rather than a single hypothesis. 
More broadly, decision-theoretic formulations arise in numerous signal processing applications, including sensor scheduling \cite{krishnamurthyStructuredThresholdPolicies2007} and beamforming design \cite{seoTrainingBeamSequence2016}. 
Yet, despite the central role of statistical decision theory, the asymptotic behavior of the optimal Bayes regret for general finite-state finite-action decision problems remains largely unexplored.

In this paper, we completely characterize the asymptotic decay rate of the optimal Bayes regret for finite-state finite-action Bayesian decision problems with conditionally i.i.d. observations. Our main result establishes that the optimal regret always decays exponentially fast and identifies the exact exponent. Specifically, the exponent is given by the minimum Chernoff information over the minimal incompatible subsets of states, where a subset is called incompatible if no single action is optimal for all states in the subset. This characterization immediately recovers the classical error exponents for both hypothesis testing and hypothesis exclusion, and further yields, to the best of our knowledge, the first exact error-exponent characterization for list hypothesis testing \cite{asadikangarshahiMinimumProbabilityError2023}.

Beyond the exponent characterization itself, our analysis uncovers a previously hidden combinatorial structure underlying statistical decision problems. We show that incompatible subsets correspond precisely to transversals of a hypergraph induced by the loss function, which we call the regret support hypergraph. Leveraging the classical bottleneck theorem \cite{edmondsBottleneckExtrema1970}, we reduce the regret analysis to a collection of hypothesis exclusion problems restricted to incompatible subsets. This connection between statistical decision theory, information-theoretic error exponents, and hypergraph combinatorics forms the technical core of our proof.

\paragraph*{Notations}
For a finite set $B$,  $\Delta_B:=\{p:B\rightarrow [0,1]: \sum_{b \in B} p(b)=1\}$ denotes the set of probability mass functions (PMFs) on $B$.
For a collection of subsets $\mathcal{S}\subset 2^B$ of a finite set $B$, $\Min(\mathcal{S})$ denotes the collection of all minimal elements of $\mathcal{S}$, that is, $\Min(\mathcal{S}):=\{S\in\mathcal{S}: \forall S'\in \mathcal{S}\backslash\{S\}, S' \not\subset S \}$.
For a set $B$ and $n \in \mathbb{N}$, $\binom{B}{n}$ denotes the collection of all subsets of $B$ with cardinality $n$.
An $n$-tuple $(x_1,\cdots,x_n) \in \mathcal{X}^n$ of elements in the same set $\mathcal{X}$ is denoted by $x^n$.
For a measure $\mu$ on a measurable space $\mathcal{X}$ and $n \in \mathbb{N}$, $\mu^{\otimes n}$ denotes the $n$-fold product measure of $\mu$ on $\mathcal{X}^n$.
Also, for a function $g:\mathcal{X} \rightarrow \mathbb{R}$, the function $g^{\otimes n}:\mathcal{X}^n \rightarrow \mathbb{R}$ is given by $g^{\otimes n}(x^n)=\prod_{i=1}^{n}g(x_i)$.
All logarithms are to base $e$.

\section{Problem Formulation}
We consider the finite-state finite-action Bayesian statistical decision problem \cite{bergerStatisticalDecisionTheory1985} $(\mathcal{H},\mathcal{A}, \ell,P,\eta)$, where
\begin{itemize}
\item $\mathcal{H}$ is a non-empty finite set of states (or hypotheses),
\item $\mathcal{A}$ is a non-empty finite set of actions,
\item $\ell:\mathcal{H}\times \mathcal{A} \rightarrow \mathbb{R}$ is a loss function,
\item $P=(P_h)_{h \in \mathcal{H}}$ is a collection of distinct probability distributions on a common measurable space $\mathcal{X}$, and
\item $\eta \in \Delta_{\mathcal{H}}$ is a prior distribution satisfying $\eta(h)>0$ for all $h \in \mathcal{H}$.
\end{itemize}
The unknown state $H$ is drawn according to $\eta$, and conditioned on $H=h$, the observations $X_1,X_2,\ldots$ are generated i.i.d. according to $P_h$. Given $n$ observations $X^n=(X_1,\ldots,X_n)$, a decision rule $f:\mathcal{X}^n\to\mathcal{A}$ selects an action $f(X^n)$.

The \textbf{(Bayes) risk} of a decision rule $f$ is  the expected loss $\mathbb{E}\left[\ell(H, f(X^n))\right]$, and the \textbf{optimal (Bayes) risk} is defined as the smallest risk over all possible decision rules,
\begin{align}
    R_n(P, \eta,\ell) := \inf_{f:\mathcal{X}^n\rightarrow \mathcal{A}} \mathbb{E}\left[\ell(H, f(X^n))\right].
\end{align}

The \textbf{regret loss function} $\ell^\reg$ is defined as
\begin{align}
    \ell^\reg(h,a):=\ell(h,a)-\ell^*(h), 
\end{align}
where $\ell^*(h):=\min_{a \in \mathcal{A}} \ell(h,a)$ is the minimum attainable loss under state $h$.  By construction, $\ell^{\reg}(h,a)\ge 0$ for all $(h,a)\in\mathcal H\times\mathcal A$, and $\min_{a\in\mathcal A}\ell^{\reg}(h,a)=0$ for every $h\in\mathcal H$. 

Moreover, for any decision rule $f$, 
\begin{align}
    \mathbb{E}\left[\ell(H, f(X^n))\right] \!= \!\!\sum_{h \in \mathcal{H}} \eta(h) \ell^*(h) +  \mathbb{E}\left[\ell^{\reg}(H, f(X^n))\right],
\end{align}
which implies  
\begin{align}
    R_n(P, \eta,\ell) = \sum_{h \in \mathcal{H}} \eta(h) \ell^*(h) +  R_n(P, \eta,\ell^{\reg}).
\end{align}
Therefore, $\ell$ and $\ell^{\reg}$ induce the same set of optimal decision rules. Accordingly, we refer to $R_n(P,\eta,\ell^{\reg})$ as the \textbf{optimal (Bayes) regret}. Our objective is to characterize the asymptotic decay rate of the optimal regret as $n\to\infty$.

\section{Preliminaries and Previous Results}
\subsection{Multivariate Chernoff Information}
The following information measure plays an important role in our analysis. 
\begin{definition}[\cite{kanayaAsymptoticsPosteriorEntropy1995, mishraOptimalErrorExponents2024}]\label{def:multChernoffInf}
    Let $\{P_h\}_{h \in \mathcal{H}}$ be a finite set of probability measures on the same measurable space $\mathcal{X}$.
    The \textbf{(multivariate) Chernoff information} is defined as
    \begin{multline}
        C(\{P_h\}_{h \in \mathcal{H}})\\
        :=-\log \inf_{s\in\Delta_{\mathcal{H}}}\int \left(\prod_{h\in\mathcal{H}}(p_h(x))^{s(h)} \right)d\mu(x),
    \end{multline}
    where $\mu$ is a measure on $\mathcal{X}$ dominating all $P_h$ and $p_h=\frac{dP_h}{d\mu}$.
    Also, in the integrand, we adopt the convention $0^0=0$.
\end{definition}

\subsection{Hypothesis Testing}
One fundamental example of a statistical decision problem is symmetric multiple hypothesis testing \cite{kanayaAsymptoticsPosteriorEntropy1995, bergerStatisticalDecisionTheory1985}. The goal is to estimate the true state so as to minimize the probability that the estimate differs from the true state. It can be formulated as a statistical decision problem with $\mathcal{A}=\mathcal{H}$ and the zero-one loss function defined as {$\ell_{\HT(\mathcal{H})}(h,a):=\mathbbm{1}(h \neq a)$}.
Note that $\ell_{\HT(\mathcal{H})}^* \equiv 0$ and $\ell_{\HT(\mathcal{H})}^\reg = \ell_{\HT(\mathcal{H})}$, and hence the optimal risk coincides with the optimal regret. 
In this case, it is well-known that \cite{kanayaAsymptoticsPosteriorEntropy1995} the optimal risk converges to $0$ exponentially fast, with exponent equal to the pairwise-minimum Chernoff information, i.e.,
\begin{multline}
    \lim_{n \rightarrow \infty}\! -\frac{1}{n} \log R_n(P, \eta,\ell_{\HT(\mathcal{H})}) \!
    =\!\! \min_{h,h'\in \mathcal{H}:h\neq h'}\! C(\{P_h, P_{h'}\}). \label{eq:errExpHT}
\end{multline}

Note that the exponent does not depend on the prior distribution $\eta$.

\subsection{Hypothesis Exclusion}
Recently, the hypothesis exclusion problem has attracted considerable attention \cite{mishraOptimalErrorExponents2024, jiConverseBoundsQuantum2025}. In contrast to conventional hypothesis testing, the goal is to select a state that is unlikely to be true. The objective is therefore to minimize the probability that the selected state is the true state. 
It can be formulated as the statistical decision problem with $\mathcal{A}=\mathcal{H}$ and the loss function $\ell_{\HE(\mathcal{H})}(h,a):=\mathbbm{1}(h = a)$.
Again, $\ell_{\HE(\mathcal{H})}^* \equiv 0$ and $\ell_{\HE(\mathcal{H})}^\reg = \ell_{\HE(\mathcal{H})}$.
It has been shown in \cite[Theorem 6]{mishraOptimalErrorExponents2024} that the optimal risk converges to $0$ exponentially fast, with exponent equal to the multivariate Chernoff information of all hypotheses.
In other words,
\begin{align}
    \lim_{n \rightarrow \infty} -\frac{1}{n} \log R_n(P, \eta,\ell_{\HE(\mathcal{H})}) = C(\{P_h\}_{h \in \mathcal{H}}). \label{eq:errExpHE}
\end{align}

Similar to the hypothesis testing problem, the exponent does not depend on the prior distribution $\eta$.

\subsection{Hypergraph Transversals and Bottleneck Theorem}\label{subsec:transversalAndBottleneckThm}
The proof of our main result relies on the notion of a \textbf{transversal} of a hypergraph \cite{brettoHypergraphTheoryIntroduction2013, edmondsBottleneckExtrema1970, diestelGraphTheory2017}, also known as a blocking set or vertex cover, and on the classical \textbf{bottleneck theorem} for transversals \cite{edmondsBottleneckExtrema1970}.
We first introduce the notion of a transversal and the associated notation.\begin{definition}
    Let $G=(\mathcal{V},\mathcal{E})$ be a hypergraph with a vertex set $\mathcal{V}$ and an edge set $\mathcal{E} \subset 2^{\mathcal{V}}$.
    A set of vertices $T \subset \mathcal{V}$ is called a \textbf{transversal} of $(\mathcal{V},\mathcal{E})$ if it intersects every edge, that is, $T \cap E \neq \emptyset$ for every $E \in \mathcal{E}$.
    The set of all transversals is denoted by $\tilde{\mathcal{T}}(G)$, and we let $\mathcal{T}(G):=\Min(\tilde{\mathcal{T}}(G))$.
\end{definition}

We next state the bottleneck theorem.
\begin{theorem}[Bottleneck theorem \cite{edmondsBottleneckExtrema1970}]\label{thm:bottleneckThm}
    Let $G=(\mathcal{V},\mathcal{E})$ be a hypergraph, where $\mathcal{V}$ is a finite set and $\mathcal{E}\subset 2^{\mathcal{V}} \backslash\{\emptyset\}$, $\mathcal{E} \neq \emptyset$.
    Then, for any function $f:\mathcal{V}\rightarrow \mathbb{R}$, we have
    \begin{align}
        \min_{E \in \mathcal{E}}\max_{v \in E} f(v) = \max_{T\in \mathcal{T}(G)} \min_{v \in T} f(v). \label{eq:bottleneckThm}
    \end{align}
\end{theorem}
The term \emph{bottleneck theorem} originates from its applications to duality results for a variety of bottleneck optimization problems, including a minimax analogue of the max-flow min-cut theorem \cite{edmondsBottleneckExtrema1970}.

\section{Main Result}
Our main result characterizes the asymptotic behavior of the optimal regret for general finite-state finite-action Bayesian decision problems. Specifically, we prove that the optimal regret decays exponentially fast to zero and identify the exact exponent governing this decay. To this end, we introduce the notion of an \textbf{incompatible} subset of states.

\begin{definition}
Let a loss function $\ell:\mathcal{H}\times\mathcal{A}\to\mathbb{R}$ be given. A subset $T\subset\mathcal H$ is said to be \textbf{$\ell$-incompatible} if no action $a\in\mathcal A$ simultaneously attains the minimum loss for all states in $T$, i.e., $\ell(h,a)=\ell^*(h), \forall h \in T$ fails  for every  $a \in \mathcal{A}$. Let  $\tilde{\mathcal{T}}_\ell$ denote the collection of all $\ell$-incompatible subsets, and define $\mathcal{T}_\ell:=\Min(\tilde{\mathcal{T}}_\ell)$.
\end{definition}
Intuitively, a subset $T$ is $\ell$-incompatible if there is no action that is simultaneously optimal for all states in $T$. Since $\ell(h,a)=\ell^*(h)$ if and only if $\ell^{\reg}(h,a)=0$, an equivalent characterization is that no action incurs zero regret for every state in $T$.

Alternatively, $\ell$-incompatible sets admit the following hypergraph interpretation. 
For each action $a \in \mathcal{A}$, define the corresponding \emph{regret support} as
\begin{align}
    S_\ell(a):=\{h \in \mathcal{H}:\ell^\reg(h,a)>0\}, \label{eq:defRegretSupp}
\end{align}
and let 
\begin{align}
    \mathcal{S}_\ell:=\{S_\ell(a): a \in \mathcal{A}\}. \label{eq:defRegretSuppCollection}
\end{align}
Define the hypergraph
\begin{align}
G_\ell:=(\mathcal{H}, \mathcal{S}_\ell),
\end{align}
which we call the \emph{regret support hypergraph}. It follows directly from the definitions that a subset of states is $\ell$-incompatible if and only if it is a transversal of $G_\ell$. Consequently,
\begin{align}
    \tilde{\mathcal{T}}_\ell = \tilde{\mathcal{T}}(G_\ell),\quad \mathcal{T}_\ell = \mathcal{T}(G_\ell).
\end{align}

\begin{remark}
 If $\mathcal{H}$ itself is not $\ell$-incompatible, there exists an action $a^*\in\mathcal A$ such that $\ell(h,a^*)=\ell^*(h)$ for every $h\in\mathcal H$.  Then, the optimal regret satisfies
$R_n(P, \eta,\ell^{\reg})=0$
for every $P$, $\eta$, and $n$, achieved by the constant decision rule $f\equiv a^*$.
\end{remark}
\begin{remark}
    No singleton subset is $\ell$-incompatible.
\end{remark}

We now state the main result. Its proof is provided in Section~\ref{sec:proofMainResult}.
\begin{theorem}\label{thm:MainThmErrExp}
   Let $(\mathcal{H},\mathcal{A},\ell,P,\eta)$ be a Bayesian decision problem, and suppose that $\mathcal{H}$ is $\ell$-incompatible. Then the optimal regret decays exponentially fast to zero as the number of observations $n$ tends to infinity. Furthermore, the optimal decay exponent is equal to the minimum Chernoff information among all minimal $\ell$-incompatible subsets. Specifically,
    \begin{align}
        \lim_{n \rightarrow \infty} -\frac{1}{n} \log R_n(P, \eta, \ell^\reg)
        =\min_{T \in \mathcal{T}_\ell} C(\{P_h\}_{h \in T}).
    \end{align}
\end{theorem}
A crucial observation is that the exponent depends only on the structure of the regret support hypergraph. In particular, it is independent of both the prior distribution $\eta$ and the specific values of the loss function, depending only on which state-action pairs incur positive regret.

\subsection{Recovery of Previous Results}
Theorem~\ref{thm:MainThmErrExp} recovers the exponent characterizations for hypothesis testing \cite{kanayaAsymptoticsPosteriorEntropy1995} and hypothesis exclusion \cite{mishraOptimalErrorExponents2024} given in  \eqref{eq:errExpHT} and \eqref{eq:errExpHE}, respectively. Specifically,
\begin{itemize}
    \item For hypothesis testing, every subset of $\mathcal{H}$ with cardinality at least two is $\ell_{\HT(\mathcal{H})}$-incompatible. 
    Hence, $\mathcal{T}_\ell=\binom{\mathcal{H}}{2}$, which recovers \eqref{eq:errExpHT}.
    \item For hypothesis exclusion, $\mathcal{H}$ is the unique  $\ell_{\HE(\mathcal{H})}$-incompatible set, i.e., $\mathcal{T}_\ell=\{\mathcal{H}\}$. Indeed, for any proper subset $T\subsetneq\mathcal{H}$, there exists an element $h\in\mathcal{H}\setminus T$, and the action $a=h$ incurs zero regret for every state in $T$. Therefore, Theorem~\ref{thm:MainThmErrExp} recovers \eqref{eq:errExpHE}.
\end{itemize}

\subsection{Application: List Hypothesis Testing}
As an interpolation between hypothesis testing and hypothesis exclusion, consider \textbf{list hypothesis testing} \cite{asadikangarshahiMinimumProbabilityError2023}. In this problem, the decision rule outputs a list of $m$ candidate states rather than a single state, where $1\le m<|\mathcal H|$. The objective is to minimize the probability that the true state is not contained in the output list. This can be formulated as the decision problem with
\begin{align}
    \mathcal{A}=\binom{\mathcal{H}}{m}, \quad \ell(h,a)=\mathbbm{1}(h \notin a).
\end{align}
The conventional hypothesis testing problem corresponds to the case $m=1$. At the other extreme, hypothesis exclusion corresponds to $m=|\mathcal H|-1$, since selecting a single unlikely state is equivalent to outputting the complementary list of $|\mathcal H|-1$ candidate states. Note that $\ell^*\equiv 0$ and hence $\ell^{\reg}=\ell$.

For this problem, a subset $T\subset\mathcal H$ is $\ell$-incompatible if and only if $|T|\ge m+1$. Therefore, $\tilde{\mathcal T}_\ell$ is the collection of all subsets of $\mathcal H$ with cardinality at least $m+1$, and hence  $\mathcal{T}_\ell=\binom{\mathcal{H}}{m+1}$. Applying Theorem~\ref{thm:MainThmErrExp}, we obtain
\begin{align}
    \lim_{n \rightarrow \infty} -\frac{1}{n} \log R_n(P, \eta, \ell)
        = \min_{T \in \binom{\mathcal{H}}{m+1}} C((P_h)_{h \in T}).
\end{align}

\section{Proof}\label{sec:proofMainResult}
In this section, we prove Theorem~\ref{thm:MainThmErrExp}. The key idea is to relate the optimal regret to the optimal risk of hypothesis exclusion on an $\ell$-incompatible subset of hypotheses through the bottleneck theorem (Theorem~\ref{thm:bottleneckThm}). We now proceed with the formal argument.
\begin{IEEEproof}[Proof of Theorem~\ref{thm:MainThmErrExp}]
    For simplicity, let us write $R_n:=R_n(P, \eta,\ell^\reg)$.
    Also, let $\mu$ be a measure on $\mathcal{X}$ dominating all $P_h$, and $p_h=\frac{dP_h}{d\mu}$.
    Moreover, for $h \in \mathcal{H}$ and $x^n \in \mathcal{X}^n$, let
    \begin{align}
        p(h, x^n) = \eta(h) p_h^{\otimes n} (x^n).
    \end{align}
    Then, we have
    \begin{multline}
        \mathbb{E}\left[\ell^\reg(H, f(X^n))\right] \\
        = \int \sum_{h \in \mathcal{H}} p(h, x^n)  \ell^\reg(h, f(x^n))  d\mu^{\otimes n}(x^n).
    \end{multline}
    From this, an optimal decision rule can be chosen to satisfy
    \begin{align}
        f(x^n) \in \argmin_{a \in \mathcal{A}} \sum_{h \in \mathcal{H}} p(h, x^n)  \ell^\reg(h, a),
    \end{align}
    and the optimal regret is given by
    \begin{multline}
        R_n = \int \left(\min_{a\in\mathcal{A}}\sum_{h \in \mathcal{H}} p(h, x^n) \ell^\reg(h,a)\right) d\mu^{\otimes n}(x^n).
    \end{multline}
    Let us define
    \begin{align}
        m &= \min\{\ell^\reg(h,a):a \in\mathcal{A},\,\, h \in S_\ell(a)\},\\
        M &= \max\{\ell^\reg(h,a):a \in\mathcal{A},\,\, h \in S_\ell(a)\}.
    \end{align}
    By the assumption that $\mathcal{H}$ is $\ell$-incompatible, each $S_\ell(a)$ is non-empty, and thus $m,M$ are well-defined and $M \geq m > 0$.
    Then, we have
    \begin{align}
        m C_n \leq R_n \leq M C_n,
    \end{align}
    where
    \begin{align}
        C_n = \int \left(\min_{a\in\mathcal{A}}\sum_{h \in S_{\ell}(a)} p(h, x^n)\right) d\mu^{\otimes n}(x^n).
    \end{align}
    Note that
    \begin{align}
        \sum_{h \in S_{\ell}(a)} p(h, x^n) &\geq \max_{h\in S_{\ell}(a)} p(h, x^n),\\
        \sum_{h \in S_{\ell}(a)} p(h, x^n) &\leq |\mathcal{H}| \max_{h\in S_{\ell}(a)} p(h, x^n).
    \end{align}
    Thus, we have
    \begin{align}
        m C'_n \leq R_n \leq |\mathcal{H}|M C'_n, \label{eq:HEstErrBound}
    \end{align}
    where
    \begin{align}
        C'_n = \int \left(\min_{a\in\mathcal{A}}\max_{h\in S_{\ell}(a)} p(h, x^n)\right) d\mu^{\otimes n}(x^n).
    \end{align}
    Now, by applying the bottleneck theorem (Theorem~\ref{thm:bottleneckThm}), we obtain
    \begin{align}
        \min_{a\in\mathcal{A}}\max_{h\in S_{\ell}(a)} p(h, x^n) = \max_{T\in\mathcal{T}_\ell}\min_{h\in T} p(h, x^n).
    \end{align}
    Hence, we have
    \begin{align}
        C'_n = \int \left(\max_{T\in\mathcal{T}_\ell}\min_{h\in T} p(h, x^n)\right) d\mu^{\otimes n}(x^n).
    \end{align}
    Note that
    \begin{align}
        C_n' \geq \max_{T\in\mathcal{T}_\ell} \int \left(\min_{h\in T} p(h,x^n) \right) d\mu^{\otimes n}(x^n)
    \end{align}
    and
    \begin{align}
        C_n' &\leq \sum_{T \in \mathcal{T}_\ell} \int \left(\min_{h\in T} p(h,x^n)\right) d\mu^{\otimes n}(x^n)\\
        &\leq 2^{|\mathcal{H}|} \max_{T\in\mathcal{T}_\ell} \int \left(\min_{h\in T} p(h,x^n) \right) d\mu^{\otimes n}(x^n).
    \end{align}
    Thus, we have
    \begin{align}
        m C_n'' \leq R_n \leq |\mathcal{H}|2^{|\mathcal{H}|} M C_n'', \label{eq:riskBoundByHE}
    \end{align}
    where
    \begin{align}
        C_n'' = \max_{T\in\mathcal{T}_\ell} \int \left(\min_{h\in T} p(h,x^n) \right) d\mu^{\otimes n}(x^n).
    \end{align}
    Now, we can observe that
    \begin{multline}
        \int \left(\min_{h\in T} p(h,x^n) \right) d\mu^{\otimes n}(x^n) \\
        = \left(\sum_{h' \in T} \eta(h')\right) R_n(P|_T, \eta|_T, \ell_{\HE(T)}),
    \end{multline}
    where 
    \begin{align}
        P|_T = (P_h)_{h \in T}, \quad \eta|_T(h) = \frac{\eta(h)}{\sum_{h' \in T} \eta(h')}.
    \end{align}
    From the previous result \cite{mishraOptimalErrorExponents2024} on hypothesis exclusion stated in \eqref{eq:errExpHE}, $R_n(P|_T, \eta|_T, \ell_{\HE(T)})$ converges to $0$ exponentially fast, with the exponent
    \begin{align}
        \lim_{n \rightarrow \infty} -\frac{1}{n}\log R_n(P|_T, \eta|_T, \ell_{\HE(T)})
        = C(\{P_h\}_{h \in T}).
    \end{align}
    Thus, $C_n''$ converges to $0$ with
    \begin{align}
        \lim_{n \rightarrow \infty} -\frac{1}{n}\log C_n'' = \min_{T \in \mathcal{T}_\ell} C(\{P_h\}_{h \in T}).
    \end{align}
    Consequently, from \eqref{eq:riskBoundByHE}, $R_n$ also converges to $0$ with
    \begin{align}
        \lim_{n \rightarrow \infty} -\frac{1}{n}\log R_n = \min_{T \in \mathcal{T}_\ell} C(\{P_h\}_{h \in T}).
    \end{align}
\end{IEEEproof}

\section{Conclusion}

We characterized the exact exponential decay rate of the optimal Bayes regret for finite-state finite-action Bayesian decision problems with conditionally i.i.d. observations. We showed that the optimal regret always decays exponentially fast and identified its exponent as the minimum multivariate Chernoff information over the minimal $\ell$-incompatible subsets of states, equivalently, the minimal transversals of the regret support hypergraph induced by the loss function. This characterization unifies the classical pairwise-minimum Chernoff exponent for symmetric multiple hypothesis testing and the multivariate Chernoff exponent for hypothesis exclusion, while also yielding an exact exponent characterization for list hypothesis testing.

An interesting direction for future work is to extend the result to quantum statistical decision problems, where each hypothesis is represented by a quantum state. A major challenge is that, unlike quantum hypothesis testing, whose exponent is characterized by the pairwise-minimum quantum Chernoff information \cite{liDiscriminatingQuantumStates2016}, no general single-letter characterization is currently known for the quantum hypothesis exclusion exponent. Another promising direction is to study the asymptotic behavior of more general information measures between states and observations. This would extend recent works \cite{wuStrongAsymptoticComposition2024, taylorAsymptoticBehaviorInformation2026}, which establish exponential convergence governed by the pairwise-minimum Chernoff information under certain structural conditions on the information measure.

\bibliographystyle{ieeetr}
\bibliography{IEEEabrv, refs}

\end{document}